\documentclass[submission,copyright,creativecommons]{eptcs}
\providecommand{\event}{GandALF 2021} % Name of the event you are submitting to
\usepackage{underscore}           % Only needed if you use pdflatex.

\usepackage[utf8]{inputenc}
\usepackage[english]{babel}
\usepackage{mathtools}
\usepackage[cal=cm]{mathalfa}
\usepackage{amssymb}
\usepackage{amsthm, thmtools}
\usepackage{caption}
\usepackage{enumerate}

\usepackage{multicol}

\usepackage{pgf, tikz}
\usetikzlibrary{automata,arrows, arrows.meta, positioning}
\usepackage{color}

\usepackage[T1]{fontenc}
\usepackage{microtype} 
\usepackage{bbold} %for mathbb 1
\usepackage{multirow} %for tabular
\usepackage{xspace}
\usepackage{setspace}
\declaretheoremstyle[spaceabove=20pt, spacebelow = 20pt, postheadspace=10pt]{myDefi}
\declaretheoremstyle[spaceabove=20pt, spacebelow = 20pt, bodyfont=\itshape, postheadspace=10pt]{myTheo}
\declaretheoremstyle[spaceabove=6pt, spacebelow = 6pt, postheadspace=10pt]{myClai}
\declaretheorem[name=Definition, style=myDefi, parent=section]{dfn}
\declaretheorem[name=Theorem, style=myTheo, sibling=dfn]{thr}
\declaretheorem[name=Lemma, style=myTheo, sibling=dfn]{lmm}
\declaretheorem[name=Corollary, style=myTheo, sibling=dfn]{crl}
\declaretheorem[name=Claim, style=myClai]{clm}
\declaretheorem[name=Example, style=myClai, sibling=dfn]{xmp}

\definecolor{gruen}{rgb}{.2,.6,.2}
\definecolor{orange}{rgb}{.9,.5,0}
\definecolor{rot}{rgb}{.8,0,0}
\definecolor{rot2}{rgb}{.8,.27,.27}
\definecolor{blau}{rgb}{0,0,.7}
\definecolor{turkis}{rgb}{0,.7,.7}
\definecolor{violett}{rgb}{0.6,0.0,.9}
\definecolor{grau}{rgb}{0.47,0.47,0.47}
\definecolor{grau2}{rgb}{0.1,0.1,0.1}

\tikzset{
main edge/.style={line width=0.7pt, draw=grau2,>={Latex[width=3pt]},-to, },
main node/.style={state,fill=white,draw=grau2, line width=0.6pt, scale=0.8},
edge label/.style={midway,text=black, scale=1,draw=none}
}

\definecolor{orange}{rgb}{.9,.5,0}
\definecolor{rot}{rgb}{.8,0,0}
\definecolor{rot2}{rgb}{.8,.27,.27}
\definecolor{blau}{rgb}{0,0,.7}
\definecolor{turkis}{rgb}{0,.7,.7}
\definecolor{violett}{rgb}{0.6,0.0,.9}
\definecolor{grau}{rgb}{0.47,0.47,0.47}
\definecolor{grau2}{rgb}{0.1,0.1,0.1}

\tikzset{
main edge/.style={line width=0.7pt, draw=grau2,>={Latex[width=3pt]},-to, },
main node/.style={state,fill=white,draw=grau2, line width=0.6pt, scale=0.8},
edge label/.style={midway,text=black, scale=1,draw=none}
}

\newcommand{\wafa}{\textup{WAFA}\xspace}
\newcommand{\wfta}{\textup{WFTA}\xspace}
\newcommand{\wfa}{\textup{WFA}\xspace}
\newcommand{\pola}{\textup{PA}\xspace}
\newcommand{\wafaa}{(Q,\Sigma,\delta,P_0,\tau)}
\newcommand{\wftaa}{(Q,\Gamma,\delta,\lambda)}
\newcommand{\N}{\mathbb{N}}
\newcommand{\A}{\mathcal{A}}
\newcommand{\B}{\mathcal{B}}
\newcommand{\Se}{\mathcal{S}}
\newcommand{\pols}[1]{S\langle #1\rangle}

\newcommand{\chr}[1]{\mathbb{1}_{#1}}
\newcommand{\bhv}[1]{[ \mkern-3mu [ #1 ] \mkern-3mu ]}
\newcommand{\bhvs}[1]{[ #1 ]}
\newcommand{\ser}[2]{#1 \langle \mkern-4mu \langle #2 \rangle \mkern-4mu \rangle}
\newcommand{\con}[1]{%
{\mathop{#1}\limits^{\vbox to -0.5\ex@{\kern-\tw@\ex@
   \hbox{\scriptsize $\sim$}\vss}}}}

\newcommand{\tw}{t_w^r}
\newcommand{\tlg}{T_\Gamma}

\DeclareMathOperator{\pos}{Pos}
\DeclareMathOperator{\ra}{r}
\DeclareMathOperator{\pot}{\mathcal{P}}
\DeclareMathOperator{\rank}{Rank}

\newcommand{\lbl}[1]{\operatorname{Label}_{#1}}

\newcommand{\oset}[2]{%
  {\mathop{#2}\limits^{\vbox to 1\ex@{\kern-\tw@\ex@
   \hbox to 14\ex@{\scriptsize $#1$}\vss}}}}
\newcommand{\cons}[1]{%
  {\mathop{#1}\limits^{\vbox to -1\ex@{\kern-\tw@\ex@
   \hbox{\scriptsize $\sim$}\vss}}}}
\def\moverlay{\mathpalette\mov@rlay}
\def\mov@rlay#1#2{\leavevmode\vtop{%
   \baselineskip\z@skip \lineskiplimit-\maxdimen
   \ialign{\hfil$\m@th#1##$\hfil\cr#2\crcr}}}
\newcommand{\charfusion}[3][\mathord]{
    #1{\ifx#1\mathop\vphantom{#2}\fi
        \mathpalette\mov@rlay{#2\cr#3}
      }
    \ifx#1\mathop\expandafter\displaylimits\fi}
\makeatother

\let\emptyset\varnothing%leereMeneg

\newcommand{\overbar}[1]{\mkern 1.5mu\overline{\mkern-1.5mu#1\mkern-1.5mu}\mkern 1.5mu}

\title{A Nivat Theorem for Weighted Alternating Automata over Commutative Semirings\thanks{This work was supported by Deutsche Forschungsgemeinschaft (DFG), Graduiertenkolleg 1763 (QuantLA).}}
\author{Gustav Grabolle
\institute{Institute of Computer Science, Leipzig University, 04109 Leipzig, Germany}
\email{grabolle at infromatik.uni-leipzig.de}
}
\def\titlerunning{Nivat Therorem for WAFA}
\def\authorrunning{G. Grabolle}
\begin{document}
\maketitle

\begin{abstract}
In this paper, we give a Nivat-like characterization for weighted alternating automata over commutative semirings (WAFA).
To this purpose we prove that weighted alternating can be characterized as the concatenation of weighted finite tree automata (WFTA) and a specific class of tree homomorphism. We show that the class of series recognized by weighted alternating automata is closed under inverses of homomorphisms, but not under homomorphisms. We give a logical characterization of weighted alternating automata, which uses weighted MSO logic for trees. Finally we investigate the strong connection between weighted alternating automata and polynomial automata. Using the corresponding result for polynomial automata, we are able to prove that the ZERONESS problem for weighted alternating automata with the rational numbers as weights is decidable.
\end{abstract}

\section{Introduction}

Non-determinism, a situation with several possible outcomes, is usually interpreted as a choice. An (existential) automaton accepts if there exists at least one successful run. Contrary to this, one can view non-determinism as an obligation. A universal automaton would accept if all possible runs are successful. While this notion of non-determinism is less prominent, it is equally natural. Allowing for the simultaneous use of existential and universal determinism leads to the concept of alternation, such as in alternating Turing machines \cite{alternation} or alternating automata on finite \cite{BRZOZOWSKI198019}, or infinite structures \cite{MULLER1987267}. States of an alternating finite automaton (AFA) are either existential, or universal. For an existential state at least one of the outgoing runs needs to be successful, for a universal state all of the outgoing runs need to be successful to make the entire run successful. It is even possible to mix both modes by assigning a propositional formula over the states to each pair of state and letter. Alternating finite automata have been known for a long time. They are more succinct than finite automata and constructions like the complement, or intersection are easy for them. Due to this, they have many uses such as a stepping stone between logics and automata \cite{DeGiacomo:2013:LTL:2540128.2540252}, or in program verification \cite{Vardi1995}.

While alternating automata recognize the same class of languages as finite automata, the situation is different in the weighted setting. A weighted finite automaton (WFA) assigns a weight to each of its transitions. The weight of a run is computed by multiplying its transition weights. Finally, the automaton assigns to each input the sum over all weights of runs corresponding to this input. By this, a weighted automaton recognizes a quantitative language a mapping from the set of words into a weight structure. Depending on the weight structure used, we may view a quantitative language, as a probability distribution over the words, as a cost or yield assignment, or as the likelihood or quantity of success for each input. To simultaneously allow for a multitude of interesting weight structures, weighted automata have been studied over arbitrary semirings \cite{Droste:2009:HWA:1667106}.

To adapt alternating automata into the weighted setting, we observe that the existence of a run in a finite automaton becomes a sum over all runs in a weighted automaton. Analogously, the demand for all runs to be successful becomes a product over all runs. More precisely, if a weighted alternating finite automaton (WAFA) is in an additive state, it will evaluate to the sum over the values of all outgoing runs. If the weighted alternating automaton is in a multiplicative state, it will evaluate to the product over the values of all outgoing runs. And again, we are able to mix both modes, this time by assigning polynomials over the states to each pair of state and letter. Weighted alternating automata over infinite words where studied in \cite{10.1007/978-3-642-03409-1_2} and in \cite{10.1007/978-3-642-24372-1_2} over finite words. While these authors focused on very specific weight structures, a more recent approach defines weighted alternating automata over arbitrary commutative semirings \cite{KOSTOLANYI20181}.

Weighted alternating automata have the same expressive power as weighted automata if and only if the semiring used is locally finite \cite{KOSTOLANYI20181}. However, for many interesting semirings such as the rational numbers, weighted alternating automata are strictly more expressive than weighted automata. While we have a fruitful framework for weighted automata, woven by results like the Nivat theorem for weighted automata \cite{droste2013weighted}, the equivalence of weighted automata and weighted rational expressions \cite{SCHUTZENBERGER1961245} and weighted restricted MSO logic \cite{DROSTE200769}, or the decidability of equality due to minimization if weights are taken from a fields \cite{SCHUTZENBERGER1961245} and many more, no such results are known for weighted alternating automata. In this paper we will extend the results on weighted alternating automata by connecting them to known formalisms and thereby establishing further characterizations of quantitative languages recognized by weighted alternating automata. From there on, we will use these connections to prove interesting properties for weighted alternating automata, but also to translate known results for weighted alternating automata into other settings.

After a brief recollection of basic notions and notations in Section 2, Section 3 will establish several normal forms for weighted alternating automata (Lemma \ref{1Lmm}, Lemma \ref{eqLmm}) that are used as the basis of later proofs. Section 4 includes our core result (Theorem \ref{thr:WafaIffWfta}), a characterization of weighted alternating automata by the concatenation of weighted finite tree automata (WFTA) together with certain homomorphisms. More precisely, we consider word-to-tree homomorphisms that translate words viewed as trees into trees over some arbitrary ranked alphabet. We can show that a quantitative language is recognized by a weighted alternating automata if, and only if there exists word-to-tree homomorphism and a weighted tree automaton such that the evaluation of the weighted alternating automata on any given word is the same as the evaluation of the weighted tree automaton on the image of the homomorphism of this word.

In Section 5 we will use this result to prove that the class of quantitative languages recognized by weighted alternating automata is closed under inverses of homomorphisms (Corollary \ref{crl:ClIHWafa}). However, we can prove the same is not true for homomorphisms in general (Lemma \ref{lmm:NClHWafa}). Since the closure under homomorphisms plays a key part in the proof of the Nivat theorem for weighted automata this prohibits a one-to-one translation of the Nivat theorem for weighted automata into the setting of weighted alternating automata. Nonetheless, we will utilize the connection between weighted alternating automata and weighted tree automata, as well as a Nivat theorem for weighted tree automata, to prove an adequate result for weighted alternating automata (Theorem \ref{thr:nivatWAFA}). This will lead us directly into a logical characterization of quantitative languages recognized by weighted alternating automata with the help of weighted restricted MSO logic for weighted tree automata (Section 6 Theorem \ref{thr:WafaIffsrMSO}). It is well known that recognizable tree languages are closed under inverses of tree homomorphisms. However, the same does not hold in the weighted setting for arbitrary commutative semirings. Section 7 gives a precise characterization of the class of semirings for which weighted tree automata are closed under inverses of homomorphisms (Theorem \ref{thr:clOTAUIHom}). For this purpose, we will use our core theorem, and a result form \cite{KOSTOLANYI20181}.

Lastly, in Section 8, we investigate the connection between weighted alternating automata and recently introduced polynomial automata \cite{8005101} to prove the decidability of the ZERONESS and EQUALITY problems for weighted alternating automata if weights are taken from the rational numbers (Corollary \ref{crl:DecZerEqWafa}).

Due to the limitation of space, we moved the technical parts of some proofs into the Appendix which can be found in the long version of this paper on the authors website.

\section{Preliminaries}
Let $\N=\{0,1,2,\ldots\}$ denote the set of non-negative integers. For sets $M, N$ we denote the cardinality of $M$ by $|M|$, the set of subsets of $M$ by $\pot(M)$, the Cartesian product of $M$ and $N$ by $M\times N$, and the set of mappings from $M$ to $N$ by $N^M=\{f \mid f:M\rightarrow N\}$. If $M$ is finite and non-empty, it is also called \textit{alphabet}.

For the remainder of this paper, let $\Sigma, \Gamma$ and $\Lambda$ denote alphabets. The set of all (finite) words over $\Sigma$ is denoted by $\Sigma^*$. Let $|w|$ denote the length of a word $w$ and $\Sigma^k = \{w\in \Sigma^*\mid |w|=k\}$. The unique word in $\Sigma^0$ is called empty word and denoted by $\varepsilon$. The concatenation of words $u,v$ is denoted by $u\cdot v$ or just $uv$. A mapping $h:\Lambda^*\rightarrow \Sigma^*$ is called \textit{homomorphism} if $h(u\cdot v)=h(u)\cdot h(v)$ and $h$ is \textit{non-deleting} if $h(a)\neq\varepsilon$ for all $a \in \Lambda$.

A \textit{monoid} is an algebraic structure $(M, \cdot, 1)$, where $\cdot$ is a binary associative internal operation and $m \cdot 1 = 1\cdot m$ for all $m \in M$. A monoid is \textit{commutative} if $\cdot$ is commutative.

A \textit{semiring} is an algebraic structure $(S,+,\cdot,0,1)$, where $(S,+,0)$ is a commutative monoid, $(S,\cdot,1)$ is a monoid, $s\cdot 0 = 0 = 0\cdot s$ for all $s \in S$, and $s_3 \cdot (s_1+s_2) = s_3\cdot s_1 + s_3\cdot s_2$ and $(s_1+s_2)\cdot s_3 = s_1\cdot s_3 + s_2\cdot s_3$ for all $s_1,s_2,s_3 \in S$. A semiring is \textit{commutative} if $(S,\cdot,1)$ is commutative.

\vspace{5pt}
\centerline{\emph{For the remainder of this paper, let $S$ denote a commutative semiring.}}
\vspace{5pt}

For any set $M$, we denote $S^{M}$ by $\ser{S}{M}$. For $L\subseteq M$ we define the \textit{characteristic function} $\chr{L}\in \ser{S}{M}$ by $\chr{L}(w)=1$ if $w\in L$ and $\chr{L}(w)=0$ otherwise for all $w\in M$. An element $s \in \ser{S}{\Sigma^*}$ is called \textit{$S$-weighted $\Sigma$-language}  (for short: weighted language). 

Let $X_n$ always denote a linear ordered set with $|X_n|=n\in \N$, we refer to the $i$-th element of $X_n$ by $x_i$. Let $\pols{X_n}$ denote the \textit{semiring of polynomials} with coefficients in $S$ and commuting indeterminates $x_1,\ldots, x_n$. We say $m \in \pols{X_n}$ is a \textit{monomial} if $m=s\cdot x_1^{k_1}\cdot \ldots \cdot x_n^{k_n}$ for some $s \in S$ and $k_1,\ldots, k_n \in \N$. The \textit{degree} of $m$ is $\sum_{i=1}^{n} k_i$. A monomial $m$ is a \textit{constant} monomial if its degree is zero. Each polynomial which can be written as a sum of distinct non-constant monomials is called a non-constant polynomial and the set of all non-constant polynomials is denoted by ${\pols{X_n}}_{\text{const}= 0}$. For $p,p_1,\ldots,p_n \in \pols{X_n}$ let $p\langle p_1,\ldots,p_n\rangle$ denote the simultaneous substitution of $x_i$ by $p_i$ in $p$ for all $1\leq i \leq n$.

A \textit{ranked alphabet} is an ordered pair $(\Gamma, \rank)$, where $\rank:\Gamma\rightarrow \N$ is a mapping. Without loss of generality, we assume $X_n \cap \Gamma=\emptyset$. Moreover, let $\Gamma^{(r)}=\{\gamma\in \Gamma\mid \rank(\gamma)=r\}$ and $\rank(\Gamma)=\max(\{\rank(\gamma)\mid \gamma \in \Gamma\})$. 

The set of \textit{$\Gamma$-terms over $X_n$} is the smallest set $T_{\Gamma}[X_n]$ such that $\Gamma^{(0)} \cup X_n\subseteq T_{\Gamma}[X_n]$; and $g(t_1,\ldots,t_{\rank(g)})\in T_{\Gamma}[X_n]$ for all $g \in \Gamma$ and all $t_1,\ldots,t_{\rank(g)}\in T_{\Gamma}[X_n]$. We denote $T_\Gamma (X_0)=T_\Gamma (\emptyset)$ by $T_\Gamma$. We extend $\rank$ by putting $\rank(x_i)=0$ for all $i\in\N$. If $\rank$ is clear from the context, we just write $g(t_1,\ldots, t_k)$. Moreover, we identify $g$ and $g()$ for $g\in \Gamma^{(0)} \cup X_n$. Hence, all terms $t\in T_\Gamma [X_n]$ are of the form $t=g(t_1,\ldots,t_k)$ for some $g\in \Gamma \cup X_n$ and $t_1,\ldots,t_k \in T_{\Gamma}[X_n]$.

We define $\pos:T_{\Gamma}[X_n]\rightarrow \pot(\N^*): g(t_1,\ldots,t_k)\mapsto \{\varepsilon\} \cup \bigcup_{i=1}^{k} \{i\}\cdot \pos(t_i)$. Let $t=g(t_1,\ldots,t_k)$. The mapping $\lbl{t}:\pos(t)\rightarrow\Gamma\cup X_n$ is defined by $\lbl{t}(\varepsilon)=g$; and $\lbl{t}(w)=\lbl{t_i}(v)$ if $w=iv\in \pos(t)$. We will identify $t$ and the mapping $\lbl{t}$: We write $t(w)$ to denote $\lbl{t}(w)$ and refer to terms as trees. Consequently, we have $t^{-1}(g) = \{w\in \pos(t) \mid \lbl{t}(w)=g\}$ for all $g\in T_{\Gamma}[X_n]$.

For $t=g(t_1,\ldots,t_k),t'\in T_\Gamma [X_n]$, and $w\in \pos(t)$, the \textit{subtree of $t$ at $w$}, denoted by $t|_w$ and the \textit{substitution of $t'$ in $t$ at $w$}, denoted by $t\langle w\leftarrow t'\rangle$ are defined by $t|_\varepsilon =t$ and $t\langle \varepsilon\leftarrow t'\rangle=t'$ if $w=\varepsilon$; and $t|_w=t_i|_v$ and $t\langle w\leftarrow t'\rangle=g(t_1,\ldots, t_{i-1},t_i\langle v\leftarrow t'\rangle, t_{i+1},\ldots, t_k)$ for $w=iv\in \pos(t)$. Moreover, let $M\subseteq \pos(t)$ and $|M|=l$. We define $t\langle M\leftarrow (t'_1,\ldots,t'_l)\rangle=t\langle m_l \leftarrow t'_l\rangle \cdots \langle m_1\leftarrow t'_1\rangle$, where $m_i$ is the $i$-th element of $M$ regarding to the lexicographical order on $\N^*$. In case $t'_1=\ldots=t'_l=t'$, we abbreviate $t\langle M\leftarrow(t'_1,\ldots,t'_l)\rangle$ by $t\langle M\leftarrow t'\rangle$. If $M=t^{-1}(x_i)$, we write $t\langle x_i\leftarrow (t'_1,\ldots, t'_l)\rangle$ to denote $t\langle M\leftarrow(t'_1,\ldots, t'_l)\rangle$. Finally, let $t\langle t'_1,\ldots,t'_n \rangle=t\langle t^{-1}(x_1)\leftarrow t'_1\rangle\cdots\langle t^{-1}(x_n)\leftarrow t'_n\rangle$ denote the simultaneous substitution in trees.

We say a tree $t$ is \textit{non-deleting in $l$ variables} if it contains at least one symbol from $\Gamma$ and each of the variables $x_1,\ldots, x_l$ occurs at least once in $t$. We say $t$ is \textit{linear in $l$ variables} if it is non-deleting in $l$ variables and each of the variables occurs at most once in $t$. Moreover, let $\ra(t)=\sum_{i=1}^n|t^{-1}(x_i)|$ and $T_{\Gamma}^{(r)}(X_n)=\{t\in T_\Gamma(X_n)\mid \ra(t)=r\}$ for all $r\in \N$.

A \textit{tree homomorphism} $h:\tlg\rightarrow T_\Lambda$ is a mapping such that for all $g \in \Gamma^{(r)}$ there exists $t_g \in T_{\Lambda}[X_r]$ with $h(g(t_1,\ldots,t_r))=t_g\langle h(t_1),\ldots, h(t_r)\rangle$ for all $t_1,\ldots,t_r \in T_\Gamma$. We will denote $t_g$ by $h(g)$, even though $t_g$ is not necessarily in $T_\Lambda$. A tree homomorphism is \textit{non-deleting} (resp. \textit{linear}) if each $h(g)$ is non-deleting (resp. linear) in $\rank(g)$ variables.

%In the following, we will consider several automata models, each of which has a set of states $Q$. If this set is finite, we assume $Q$ to be ordered by a linear order $\preceq$. Regarding to $\preceq$, we will refer to the $i$-th element of $Q$ as $q_i$.

\section{Weighted alternating finite automata}

This section introduces weighted alternating finite automata (\wafa) and shows how to achieve desirable normal forms of \wafa (Lemma \ref{1Lmm}, Lemma \ref{eqLmm}). We will follow the definitions of \cite{KOSTOLANYI20181}.

A \textit{weighted alternating finite automaton} (\wafa) is a 5-tuple $\A=\wafaa$, where $Q=\{q_1,\ldots ,q_n\}$ is a finite set of states, $\Sigma$ is an alphabet, $\delta: Q \times \Sigma \rightarrow \pols{Q}$ is a transition function, $P_0 \in \pols{Q}$ an initial polynomial, and $\tau: Q \rightarrow S$ a final weight function.

Let $\A=\wafaa$ be a \wafa. Its \textit{state behavior} $\bhvs{\A}: Q\times\Sigma^*\rightarrow S$ is the mapping defined by \[\bhvs{A}(q,w) = \begin{cases} \tau(q) &\text{if } w=\varepsilon , \\ \delta(q,a)\big\langle\bhvs{\A}(q_1,v),\ldots, \bhvs{\A}(q_n,v)\big\rangle &\text{if } w=av \text{ for } a\in \Sigma\enspace. \end{cases}\]
Usually, we will write $\bhvs{\A}_q(w)$ instead of $\bhvs{\A}(q,w)$. 
Now, the \textit{behavior} of $\A$ is the weighted language $\bhv{\A}:\Sigma^*\rightarrow S$ defined by \[\bhv{A}(w)= P_0\big\langle\bhvs{\A}_{q_1}(w),\ldots,\bhvs{\A}_{q_n}(w)\big\rangle\enspace.\]
A weighted language $s$ is \textit{recognized} by $\A$ if and only if $\bhv{\A}=s$. Two \wafa are said to be \textit{equivalent} if they recognize the same weighted language. To ease late proofs, let $M_{(q,a)}$ denote the set of monomials that appear in $\delta(q,a)$.

We say a \wafa $\wafaa$ with $Q=\{q_1,\ldots, q_n\}$ is a \textit{weighted finite automaton} (\wfa) if $P_0 = \sum_{j=1}^{n} s_j\cdot q_j$ and $\delta(q_i,a)= \sum_{j=1}^{n} s^{a}_{ij}\cdot q_j$ for all $1\leq i \leq n$, $a \in \Sigma$. This definition coincides with the usual definition \cite{HWAc5},\cite{DROSTE200769}. We can see this by defining the initial weight function by $\lambda(q_i) = s_i$, the transition weight function by $\mu(a)(q_i,q_j)=s^{a}_{ij}$, and the final weight function by $\gamma(q_i)=\tau(q_i)$ for all $1\leq i,j\leq n$ and all $a\in \Sigma$.

We say $\A$ is \textit{nice} if it has the following properties:
\begin{enumerate}[(i)]
\item $\delta(q,a)$ is a finite sum of pairwise distinct, monomials of the form $s\cdot q_1^{k_1}\cdot\ldots \cdot q_n^{k_n}$ for all $q\in Q, a\in \Sigma$,
\item all monomials in $P_0$ and $\delta$ are non-constant,
\item $P_0 =q_1$.
\end{enumerate}
Moreover, we say that $\A$ is \textit{purely polynomial} if: 
\begin{enumerate}
  \item[(iv)] all monomials (in $P_0$ and $\delta$) have coefficient $1$.
\end{enumerate}

We want to show that we can always assume a \wafa to be nice and purely polynomial.

\begin{lmm}\label{1Lmm} For each \wafa $\A$ there exists an equivalent \wafa $\A'$ such that (i)-(iv) hold for $\A'$.
\end{lmm}

\begin{proof}
Let $\A=\wafaa$ be a \wafa.
\begin{enumerate}[(i)]
\item Since $\cdot$ is distributive and commutative in $\pols{Q}$ there exists an equivalent \wafa $\A'$ such that (i) holds. 

\item Assume (i) holds for $\A=\wafaa$. We define a \wafa $\A'=(Q',\Sigma, \delta', P_0', \tau')$ which includes a new state $q_c$ for each constant $c$ occurring in $\A$. Furthermore, $\delta'$ and $P_0'$ are as $\delta$ and $P_0$, respectively, but each occurrence of each constant $c$ is replaced by $q_c$. Moreover, $\delta'(q_c,a)=q_c$ for all $a\in \Sigma$ and $\tau'(q_c)=c$. There is a finite number of constants in $\A$. Thus, $\A'$ is a \wafa.  It is easy to see that $\A$ and $\A'$ are equivalent and that (i)-(ii) hold for $\A'$.

\item Assume (i)-(ii) hold for $\A$. Due to Lemma 6.3 of \cite{KOSTOLANYI20181}, there exists an equivalent \wafa $\A'$ such that (i)-(iii) hold for $\A'$.

\item Assume that (i)-(iii) hold for $\A$. We define \[Q'=Q \cup \{q_s \mid s \text{ is the coefficient of some monomial in }\A\}.\] Without loss of generality, we can assume that these sets are disjoint. Furthermore, let $\delta'$  be defined by
\[
\delta'(q,a) = \begin{cases}\sum\limits_{s\cdot{q_1}^{k_1} \ldots {q_n}^{k_n} \in M_{(q,a)}} {q_1}^{k_1} \ldots {q_n}^{k_n} \cdot q_s &\text{if } q \in Q \text{ and}\\ q &\text{otherwise} \end{cases}\] for all $q \in Q'$, $a \in \Sigma$. Moreover, let $\tau'$  be defined by
\[
\tau'(q) = \begin{cases}\tau(q) &\text{if } q \in Q \text{ and}\\ s &\text{if } q=q_s\end{cases}
\] for all $q \in Q'$. Consider the \wafa $\A'=(Q',\Sigma,P_0,\delta',\tau')$. It is easy to see that $\bhv{\A}=\bhv{\A'}$ and that every monomial in $\A'$ has $1$ as coefficient. 

If the appropriate order on $Q'$ is chosen, properties (i)-(iii) hold for $\A'$, too. However, if this construction is applied to a \wfa $\A$, the resulting $\A'$ does not have to be a \wfa.
\end{enumerate} 
\end{proof}

In \cite{KOSTOLANYI20181} the transition function and the initial polynomial are not allowed to contain constants. This corresponds to the property that runs are not allowed to terminate before the entire word is read. Since it will ease later constructions, we allowed constants in our definition. Nevertheless, as Lemma \ref{1Lmm} (ii) shows, the introduction of constants does not increase expressiveness since it is possible to simulate terminating transitions by \glqq deadlock\grqq-states.

We say a \wafa $\wafaa$ is \textit{equalized} if all monomials occurring in $\delta$ have the same degree.

\begin{lmm}\label{eqLmm} For each \wafa $\A$ there exists an equivalent and equalized \wafa $\A'$.
\end{lmm}

\begin{proof} Let $\A=(Q,\Sigma,\delta,q_1,\tau)$ be a nice \wafa and $d$ the maximum degree of monomials occurring in $\delta$. Let $q_{n+1}$ be a new state and 
\[
\delta'(q,a) = \begin{cases}\sum\limits_{s{q_1}^{k_1} \ldots {q_n}^{k_n} \in M_{(q,a)}} s\cdot {q_1}^{k_1} \ldots {q_n}^{k_n} \cdot q_{n+1}^{d-\sum_{u=1}^n k_u} &\text{if } q \in Q \text{ and}\\ q_{n+1}^d &\text{otherwise} \end{cases}\] for all $q \in Q'$, $a \in \Sigma$. As well as, \[\tau'(q) = \begin{cases}\tau(q) &\text{if } q \in Q \text{ and}\\ 1 &\text{if } q=q_{n+1}\end{cases}
\] for all $q \in Q'$. Clearly, $\A'=(Q\cup\{q_{n+1}\},\Sigma,\delta',q_1,\tau')$ is equalized. Also, it is easy to see that $\A$ and $\A'$ are equivalent and that $\A'$ is nice.

Please note that the coefficients of monomials in $\A$ where not changed, thus we also can assume that all these coefficients are $1$.
\end{proof}

Nice \wafa can be represented in the following way: As usual we depict each state by a circle. Then, each monomial $s\cdot q_1^{k_1}\ldots q_n^{k_n}$ in $\delta(q_i,a)$ is represented by a multi-arrow which is labeled by $a:s$, begins in $q_i$, and has $k_j$ heads in $q_j$ for all $1\leq j\leq n$, respectively. In case a multi-arrow has more than one head, we join these heads by a \tikz{\draw[draw= grau2, fill=grau2] (0,0) circle (1.5pt);}. If $s=1$, we omit the $s$-label. If $s=0$, we omit the complete multi-arrow. The initial polynomial is represented analogously. The final weights are represented as usual. Note that the multi-arrows can be viewed as a parallel or simultaneous transitions and that this representation coincides with the usual representation if the automaton is a $\wfa$. Consider the following example:

\begin{xmp}\label{xmp:serp} Let $S=(\N,+,\cdot,0,1)$, $\Sigma=\{a,b\}$, and $s$ the weighted language \[\begin{array}{llcl} s: &\Sigma^*&\rightarrow &S:\\
&w &\mapsto &\begin{cases}{(2^{j})}^{2^i} &\text{if } w=a^ib^j \enspace, \\ 0 &\text{otherwise.}\end{cases}\end{array} \]
We consider the \wafa $\A=(\{q,p\},\Sigma,P_0,\delta,\tau)$, defined by:
\[\begin{array}{lcllclclcllcl}
P_0 &=&q & & & &\hspace{2cm} &\delta(q,a)&=&{q}^2 &\delta(q,b)&= &p\\
\tau(q)&=&1 &\tau(p)&=&2 & &\delta(p,a)&=&0 &\delta(p,b)&=&2\cdot p
\end{array}\]
A depiction of this automaton can be seen in Figure \ref{rep1}. One can check that $\bhv{\A}=s$, for example:
\[\begin{array}{clclcl}
& \bhv{\A}(aabb) & = & q\big\langle\bhvs{\A}_q(aabb),\bhvs{\A}_p(aabb)\big\rangle & & \\
= & \bhvs{\A}_q(aabb) & = & q^2\big\langle\bhvs{\A}_q(aabb),\bhvs{\A}_p(aabb)\big\rangle & & \\
= & \big(\bhvs{\A}_q(abb)\big)^2 & = & \big(\bhvs{\A}_q(bb)\big)^{2\cdot 2} & = & \big(\bhvs{\A}_p(b)\big)^{2\cdot 2}\\
= & \big(2\cdot\bhvs{\A}_p(\varepsilon)\big)^{2\cdot 2} & = & \big(2\cdot\tau(p)\big)^{2\cdot 2} & = & {(2^2)}^{2^2}
\end{array}\]
\begin{figure}
\centering
\begin{minipage}{.5\textwidth}
\centering
\begin{tikzpicture}[every edge/.style={main edge},
          every loop/.style={main edge}, 
          every initial by arrow/.style={main edge, initial distance= 10pt},
          every accepting by arrow/.style={main edge, accepting distance= 10pt}]
  \def\x{1.5}
  \def\y{1.5}

  \node[main node,initial by arrow, initial where=above, initial text=,accepting by arrow, accepting where= below, accepting text=$1$ ] (q1) at (-1*\x,0*\y) {$q$};
  \node[main node,accepting by arrow, accepting where= right, accepting text=$2$ ] (q2) at (1*\x,0*\y) {$p$};
  \phantom{\node[main node,accepting by arrow, accepting where= below, accepting text=$1$ ] (q3) at (0*\x,-.8*\y) {$h_1$};}

  \path (q1) edge node[edge label, above] {$b$} (q2)
             edge[out=180, in=150, looseness=9] node[edge label, xshift=-5pt, yshift=-5pt] {$a$} (q1)
             edge[out=180, in=210, looseness=9] (q1)

        (q2) edge[out=-90, in=-60, looseness=9] node[edge label, yshift=-5pt] {$b$:$2$} (q2);
  \draw[draw= grau2, fill=grau2] ([rotate around={0:(q1)}]q1.west) circle (1.5pt);

\end{tikzpicture}
\caption{Representation of $\A$}\label{rep1}
\end{minipage}%
\begin{minipage}{.5\textwidth}
\centering
\begin{tikzpicture}[every edge/.style={main edge},
          every loop/.style={main edge}, 
          every initial by arrow/.style={main edge, initial distance= 10pt},
          every accepting by arrow/.style={main edge, accepting distance= 10pt}]
  \def\x{1.5}
  \def\y{1.5}

  \node[main node,initial by arrow, initial where=above, initial text=,accepting by arrow, accepting where= below, accepting text=$1$ ] (q1) at (-1*\x,0*\y) {$q$};
  \node[main node,accepting by arrow, accepting where= right, accepting text=$2$ ] (q2) at (1*\x,0*\y) {$p$};
  \node[main node,accepting by arrow, accepting where= below, accepting text=$1$ ] (q3) at (0*\x,-.8*\y) {$h_1$};

  \path (q1) edge node[edge label, above] {$b$} (q2)
             edge[out=180, in=150, looseness=9] node[edge label, xshift=-5pt, yshift=-5pt] {$a$} (q1)
             edge[out=180, in=210, looseness=9] (q1)
             edge[out=0, in = 120] (q3)

         (q3) edge[out=180, in=150, looseness=9] node[edge label, xshift=-5pt, yshift=-5pt] {$\Sigma$} (q3)
         edge[out=180, in=210, looseness=9] (q3)

        (q2) edge[out=-90, in=-60, looseness=9] node[edge label, yshift=-5pt] {$b$:$2$} (q2)
             edge[out=-90, in=60] (q3);
  \draw[draw= grau2, fill=grau2] ([rotate around={0:(q1)}]q1.west) circle (1.5pt);
  \draw[draw= grau2, fill=grau2] ([rotate around={0:(q1)}]q1.east) circle (1.5pt);
  \draw[draw= grau2, fill=grau2] ([rotate around={0:(q3)}]q3.west) circle (1.5pt);
  \draw[draw= grau2, fill=grau2] ([rotate around={0:(q2)}]q2.south) circle (1.5pt);

\end{tikzpicture}
\caption{Representation of equalized, nice $\A$}\label{rep2}
\end{minipage}
\end{figure}
\end{xmp}

It is easy to see that $s$ from Example \ref{xmp:serp} is not recognizable by a \wfa. Thus, \wafa are more expressive than \wfa when weights are taken from the non-negative integers. However, this is not the case for every semiring. A semiring $S$ is \textit{locally finite} if for every finite $X\subseteq S$ the generated subsemiring $\langle X \rangle$ is finite.
%\[\langle X \rangle = \{\sum_{i=1}^{n} \prod_{j=1}^{m_i}x_{i,j}\in S \mid \text{for some } n, m_i \in \N \text{ and } x_{i,j} \in X\}\] is finite.
The following result characterizes semirings on which \wafa and \wfa are equally expressive:

\begin{thr}[Theorem 7.1 in \cite{KOSTOLANYI20181}]\label{thr:WAFAiffWFAiffLocFin} The class of $S$-weighted $\Sigma$-languages recognizable by \wafa and the class of $S$-weighted $\Sigma$-languages recognizable by \wfa are equal if and only if $S$ is locally finite.
\end{thr}

\section{A characterization of \wafa via weighted finite tree automata}
Our central result Theorem \ref{thr:WafaIffWfta} is included in this section, as well as the definition for weighted finite tree automata.

The connection between alternating automata and trees is well known. Often trees are used to define the runs of alternating automata. This is possible for \wafa too (see Appendix in the long version). We want to strengthen this connection by the use of tree automata and tree homomorphisms. In order to do so, we need some additional definitions.

An element $r\in \ser{S}{\tlg}$ is called \textit{($S$-weighted) tree language}. A \textit{weighted finite tree automaton} (\wfta) is a 4-tuple $\A=\wftaa$, where $Q=\{q_1,\ldots,q_n\}$ is a finite set of states, $\Gamma$ is a ranked alphabet, $\delta=(\delta_k\mid 1\leq k\leq\rank(\Gamma))$ is a family of transition functions $\delta_k:\Gamma^{(k)}\rightarrow S^{Q^k\times Q}$, and $\lambda: Q \rightarrow S$ a root weight function. 

If $k$ is clear from the context, we will denote tuples $(p_1,\ldots,p_k)$ by $\overbar{p}$. Moreover, since $k$ in $\delta_k(g)$ is clear from $g$, we will denote $\delta_k(g)$ by $\delta_g$.

Let $\A=\wftaa$ be a \wfta. Its \textit{state behavior} $\bhvs{\A}: Q\times \tlg\rightarrow S$ is the mapping defined by \[\bhvs{A}\big(q,g(t_1,\ldots,t_k)\big) =\sum\limits_{\overbar{p}\in Q^k}\delta_g(\overbar{p},q) \cdot \prod_{i=1}^{k} \bhvs{\A}(p_i,t_i)\enspace.\]
Usually, we will write $\bhvs{\A}_q(t)$ instead of $\bhvs{\A}(q,t)$. 
Now, the \textit{behavior} of $\A$ is the weighted tree language $\bhv{\A}:\tlg\rightarrow S$ defined by \[\bhv{A}(t)= \sum_{i=1}^{n}\lambda(q_i)\cdot \bhvs{\A}_{q_i}(t)\enspace.\]
A weighted tree language $s$ is \textit{recognized} by $\A$ if and only if $\bhv{\A}=s$.

It is well known that a word over $\Sigma$ can be represented as a $1$-ary tree: Each letter of $\Sigma$ is given rank one and a new end-symbol $\#$ of rank zero is added. Then $w_0w_1\ldots w_n$ translates to the tree $w_0(w_1(\ldots w_n(\#)\ldots))$. Here, we want to represent words as full $r$-ary trees for any arbitary $r \in \N$. Given an alphabet $\Sigma$ and $r\geq 1$, we define the ranked alphabet $\Sigma_\#^{r}=\Sigma \cup \{\#\}$ with $\rank(\#)=0$ and $\rank(a)=r$ for all $a \in \Sigma$. For all $w\in \Sigma^*$ the tree $\tw\in T_{\Sigma_\#^r}$ is defined by $t^r_\varepsilon=\#$; and $t^r_w=a(t^r_v,\ldots,t^r_v)$ if $w=av$ with $a\in \Sigma$. We call $h^r:\Sigma^* \rightarrow T_{\Sigma_\#^r}: w\mapsto t_w^r$ the \textit{generic tree homomorphism (of rank $r$)}. The case $r=1$ is special since for all $t\in T_{\Sigma_\#^1}$ there exists $w\in \Sigma^*$ such that $t=t_w^1$. Therefore, if clear from the context, we will identify $\Sigma$ and $\Sigma_\#^1$, $\Sigma^*$ and $T_{\Sigma_\#^1}$, as well as $w$ and $t_w^1$. It is well known that a weighted $\Sigma$ language is recognizable by a \wfa over $\Sigma$ if and only if it is recognizable by a \wfta over $\Sigma_\#^1$.

The key observation is that the behavior of a \wafa $\A$ on $w$ can be characterized by the behavior of a \wfta on $\tw$ where $r$ is the degree of polynomials in an equalized version of $\A$. Even more, the behavior of a \wfta on $h(w)$ (where $h$ is a tree homomorphism) can be characterized by the behavior of a \wafa on $w$.

\begin{lmm}\label{lmm:d1WafaIffWfta} If $s\in \ser{S}{\Sigma^*}$ is recognized by a \wafa, then $s=\bhv{\B}\circ h^r$ for some \wfta $\B$ and $r\in \N$.
\end{lmm}

\begin{proof} Assume $s$ is recognized by a \wafa. Due to Lemma \ref{1Lmm} and Lemma \ref{eqLmm}, we may assume that $s$ is recognized by a nice and equalized \wafa $\A=(Q,\Sigma,\alpha,P_0,\tau)$. Let $r$ be the unique degree of monomials in $\A$. We define the \wfta $\B=(Q,\Sigma^r_{\#},\beta,\lambda)$ with $\lambda=\chr{\{q_1\}}$ and
\[
\begin{array}{lcl}
\beta_\#(\varepsilon,q)&=&\tau(q)\enspace ,\\
\beta_a(\overbar{p},q)&=&\begin{cases}s &\text{if } s\cdot p_1\cdot\ldots\cdot p_{r}\in M_{(q,a)}\enspace,\\
                        0 &\text{otherwise}\enspace .\end{cases}
\end{array}
\]
Please note that the order of $p_1, \ldots, p_r$ is of importance. Thus, due to the fact that $\A$ is nice, $\beta_a(\overbar{p},q)$ becomes always zero if the $p_i$ are not ordered according to the linear order on the states of $\A$.

By induction over the length of $w \in \Sigma^*$ (see Appendix in the long version), we get \[\bhvs{\A}_q(w)=\bhvs{\B}_q(\tw)\] for all $q\in Q$, $w\in \Sigma^*$. Since $\A$ is nice and thus $P_0=q_1$, we consequently have \[\bhv{\A}(w) = \bhvs{\A}_{q_1}(w) = \bhvs{\B}_{q_1}(\tw) = \sum_{i=1}^{n}\lambda(q_i)\cdot \bhvs{\B}_{q_i}(\tw) = \bhv{\B}(\tw)\] for all $w\in \Sigma^*$. Since $h^r(w)=t_w^r$, this finishes our proof.
\end{proof}

The following example illustrates this connection between \wafa and \wfta.

\begin{xmp} We consider the automaton $\A$ from Example \ref{xmp:serp}. It is easy to construct the corresponding \wfta $\B=(Q,\Gamma,\beta,\lambda)$ from the equalized version $\A'$ (Figure \ref{rep2}). First, we copy the set of states (in order) $Q=\{q,p,h_1\}$. Since the maximum degree of polynomials in $\A$ was $2$ we get $\Gamma=\{a^{(2)}, b^{(2)},{\#}^{(0)}\}$. The root weight function corresponds to the initial weights. However, $\A'$ is nice and thus $\lambda = \chr{\{q\}}$. The transition weight functions $\beta_a$ and $\beta_b$ can be defined using the multi arrows in Figure \ref{rep2}. For example, the $b$-labeled multi arrow in the middle corresponds to $\beta_b(ph_1,q)=1$. Finally, the final weights in $\A'$ are captured by $\beta_\# (\varepsilon,q)=1, \beta_\# (\varepsilon,h_1)=1$, and $\beta_\# (\varepsilon,p)=2$. The only non-zero run on $t^2_{ab}$ can be seen in Figure \ref{rep3}.

\begin{figure}
\centering
\begin{tikzpicture}
    \node (hash) at (-5,-1.8) {\textbf{\#}};
    \node (b) at (-3,-1.8) {\textbf{b}};
    \node (a) at (-1,-1.8) {\textbf{a}};
    
    \node (root) at (2,0) {$1$};
    \node (1) at (0,0) {$q$};
    \node (2) at (-2,0.7) {$q$};
    \node (3) at (-2,-0.7) {$q$};
    \node (4) at (-4,1.05) {$p$};
    \node (5) at (-4,0.35) {$h_1$};
    \node (6) at (-4,-0.35) {$p$};
    \node (7) at (-4,-1.05) {$h_1$};
    \node (8) at (-6,1.05) {$2$};
    \node (9) at (-6,0.35) {$1$};
    \node (10) at (-6,-0.35) {$2$};
    \node (11) at (-6,-1.05) {$1$};

    \path (11) edge[->] (7)
          (10) edge[->] (6)
          (9) edge[->] (5)
          (8) edge[->] (4)
          (7) edge[->] node[midway, below] {$2$} (3)
          (6) edge[->] (3)
          (5) edge[->] node[midway, below] {$2$} (2)
          (4) edge[->] (2)
          (3) edge[->] node[midway, below] {$1$}(1)
          (2) edge[->] (1)
          (1) edge[->] (root);
\end{tikzpicture}
\caption{Run of translated \wfta on $t^2_{ab}$}\label{rep3}
\end{figure}
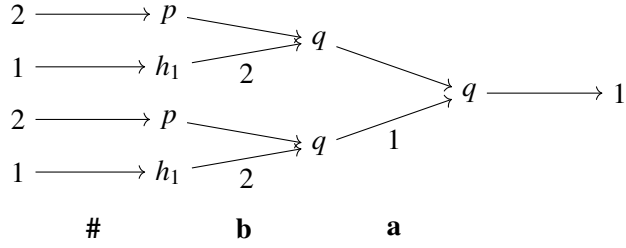
\end{xmp}

\begin{lmm}\label{lmm:d2WafaIffWfta} Let $\B=\wftaa$ a \wfta and $h:\Sigma^*\rightarrow T_\Gamma$ a tree homomorphism, then $\bhv{\B}\circ h$ is recognized by a \wafa.
\end{lmm}

\begin{proof}
Assume $s=\bhv{\B}\circ h \in \ser{S}{\Sigma^*}$, where $\B=\wftaa$ is a \wfta with $|Q|=n$ and $h: \Sigma^* \rightarrow \tlg$ a tree homomorphism. We want to construct a \wafa $\A$ such that $\bhv{\A}=s$.

If $h$ would be the generic homomorphism, we could define $\delta'(q,a)=\sum_{\overbar{p}\in Q^{r}} \delta_a(\overbar{p},q)$ and use the same proof as in the first direction. However, we want to prove this for arbitrary homomorphisms. To achieve this, we give some additional definitions.

Under $h$, each letter becomes a tree. Nonetheless, we are not interested in the structure of $h(a)$, but want to handle it as if it is a ranked letter. Therefore, we use $h(a)\langle x_1\leftarrow(p_1,\ldots,p_{r})\rangle$ to disambiguate its $r=\ra(h(a))$ variables. Furthermore, we extend the family of transition functions $(\delta_k)_{0\leq k\leq\rank{\Gamma}}$ into a family $(\delta'_k)_{k\in \N}$ with $\delta'_k:T_{\Gamma}^{(k)}(\{x_1\})\rightarrow S^{Q^k\times Q}$. We use the same notations for $\delta'$ as for $\delta$ and define $\delta'_k$ recursively as follows. 

\begin{enumerate}
  \item For all $g\in \Gamma^{(0)}$ let $\delta'_g(\varepsilon,q) =  \delta_g(\varepsilon,q)$ for all $q\in Q$,
  \item $\delta'_{x_1}(p,q) =  \chr{\{q\}}(p)$ for all $(p,q)\in Q\times Q$, and
  \item if $t=g(t_1,\ldots,t_k)$ for some $g\in \Gamma^{(k)}$, we define \[\delta'_t(\overbar{p}_1,\ldots,\overbar{p}_k,q) =\sum\limits_{\overbar{p}'\in Q^k}\delta_g(\overbar{p}',q) \cdot\prod_{i=1}^k \delta'_{t_i}(\overbar{p}_i, p'_i)\enspace\]for all $(\overbar{p}_1,\ldots,\overbar{p}_k,q)\in Q^{\ra(t_1)}\times\ldots\times Q^{\ra(t_k)}\times Q$.
\end{enumerate}Please note, for $t\in T_{\Gamma}$ we have $\delta'_t(\varepsilon,q)=\bhvs{\A}_q(t)$ by definition.

Now, we are well equipped to define the \wafa $\A=(Q,\Sigma, \alpha,P_0,\tau)$. Let $P_0=\sum_{i=1}^n \lambda(q_i)\cdot q_i$, $\tau(q_i)=\delta'_{h(\#)}(\varepsilon,q_i)$ for all $1\leq i\leq n$, and
$\alpha(q,a)=\sum_{\overbar{p} \in Q^{\ra(h(a))}} \delta'_{h(a)}(\overbar{p},q)\cdot \prod_{i=1}^{\ra(h(a))}p_i$ for all $q\in Q$, $a \in \Sigma$.

By induction over the length of $w \in \Sigma^*$ we get (see Appendix in the long version) \[\bhvs{\B}_q(h(w))=\bhvs{\A}_q(w)\] for all $q\in Q$, $w\in \Sigma^*$. Finally, for all $w\in \Sigma^*$ we get: \[\bhv{\B}(h(w)) = \sum_{i=1}^n \lambda(q_i) \cdot \bhvs{\B}_{q_i}(h(w))
= \sum_{i=1}^n \lambda(q_i) \cdot \bhvs{\A}_{q_i}(w)
= P_0\langle\bhvs{\B'}_{q_1}(w),\ldots,\bhvs{\A}_{q_n}(w)\rangle
= \bhv{\A}(w)\]
\end{proof}

This leads us to our main result.

\begin{thr} \label{thr:WafaIffWfta} A weighted language $s\in \ser{S}{\Sigma^*}$ is recognized by a \wafa if and only if there exists a ranked alphabet $\Gamma$, a tree homomorphism $h: \Sigma^* \rightarrow \tlg$, and a \wfta $\A=\wftaa$ such that $s=\bhv{\A}\circ h$.
\end{thr}

\begin{proof} This is an immediate consequence of Lemma \ref{lmm:d1WafaIffWfta} and Lemma \ref{lmm:d2WafaIffWfta}.
\end{proof}

This result allows us to transfer results from \wfta to \wafa. Moreover, additional observations in the proofs show that one can give a weight preserving, bijective mapping between the runs of $\A$ and $\B$. This allows us to translate results about runs of \wfta into results of runs of \wafa. 

% \begin{crl} Non-weighted alternating automata and finite automata are equally expressive.
% \end{crl}

% This result goes back to \cite{alternation}. However, it is well known that the non-weighted setting corresponds to the weighted setting, if weights are taken from the Boolean semiring. We observe that $\chr{L} \circ h=\chr{h^{-1}(L)}$ for any tree language $L\subseteq T_\Lambda$ and tree homomorphism $h:T_\Gamma\rightarrow T_\Lambda$. Assume $\chr{L}$ is recognized by a WAFA. By Lemma \ref{lmm:d1WafaIffWfta}, we get $\chr{L}=\chr{L'}\circ h = \chr{h^{-1}(L)}$ with $h=h^r$ for some $r\in \N$ and $L'\subseteq T_{\Sigma^r_\#}$ regular. Since regular tree languages are closed under inverses of homomorphisms and since \wfa and \wfta are equally expressive over $\Sigma_\#^1$, we get that $\chr{L}$ is recognized by a \wfa. The authors in \cite{alternation} also show that the translation of an alternating finite automaton into an equivalent deterministic finite automaton leads to (worst case) a doubly exponential blowup in states. By our proof we can see that one exponent is needed to determinize $\B$ (from the Proof of Lemma \ref{lmm:d1WafaIffWfta}) and one exponent is needed to determinize the nondeterministic finite automaton recognizing $\bhv{B}\circ h$ (the construction of the automaton recognizing $\bhv{B}\circ h$ itself is linear in states).

\section{A Nivat theorem for \wafa}
This section leads to the Nivat-like characterization of \wafa (Theorem \ref{thr:nivatWAFA}), but first we will prove that weighted languages recognized by \wafa are closed under inverses of homomorphisms (Corollary \ref{crl:ClIHWafa}), but not under homomorphisms (Lemma \ref{lmm:NClHWafa}).

Let $s_1\odot s_2$ denote the \textit{Hadamard product (pointwise product)} of two weighted languages $s_1,s_2\in \ser{S}{\Sigma^*}$. Furthermore, a word homomorphism $h:\Gamma^*\rightarrow \Sigma^*$ is called \text{non-deleting} if and only if $h(a)\neq \varepsilon$ for all $a \in \Sigma$. Let $r \in \ser{S}{\Gamma^*}$. For $h:\Gamma^*\rightarrow \Sigma^*$ a non-deleting homomorphism, we define $h(r)\in \ser{S}{\Sigma^*}$ by \(h(r)(w)=\sum_{v \in h^{-1}(w)} r(v)\) for all $w \in \Sigma^*$. For $h:\Sigma^*\rightarrow\Gamma^*$ we define $h^{-1}(r)\in \ser{S}{\Sigma^*}$ by \(h^{-1}(r)(w) = r(h(w))\) for all $w \in \Sigma^*$. Note that we have $h(\chr{L}) (w)=1$ if and only if there exists $v\in \Gamma^*$ with $h(v)=w$ and $v\in L$ in the boolean setting. Thus, $h(\chr{L})=\chr{h(L)}$. Analogously, we get $h^{-1}(\chr{L}) = \chr{h^{-1}(L)}$. Hence, $h(r)$ corresponds to the application of a homomorphism, while $h^{-1}(r)$ corresponds to the application of the inverse of a homomorphism in the non-weighted setting.

The original Nivat Theorem \cite{Nivat} characterizes word-to-word transducers. A generalized version for \wfa over arbitrary semirings (Theorem 6.3 in \cite{droste2013weighted}) can be stated in the following way:

\begin{thr}[Nivat-like theorem for \wfa \cite{droste2013weighted}]\label{NfWFA} A weighted language $s \in \ser{S}{\Sigma^*}$ is recognized by a \wfa if and only if there exist an alphabet $\Gamma$, a non-deleting homomorphism $h:\Gamma^*\rightarrow \Sigma^*$, a regular language $L \subseteq\Gamma^*$, and a \wfa $\A_w$ with exactly one state such that:\[s = h(\bhv{\A_w} \odot \chr{L}) \enspace.\]
\end{thr}

Please note, $\A_w$ does not depend on any input and is called $\A_w$ since it is responsible for the application of weights. Our goal is, to generalize this result up to \wafa. This Nivat-like theorem is strongly connected to the closure of weighted languages recognized by \wfa under (inverses) of homomorphisms. Thus, we will investigate these properties for \wafa.

\subsection{Closure properties}
A class $K$ of $S$-weighted languages is said to be \textit{closed under homomorphisms} if $s \in \ser{S}{\Sigma^*}\cap K$  and $h:\Gamma^*\rightarrow \Sigma^*$ a homomorphism implies $h(s) \in K$. Moreover, $K$ is \textit{closed under inverses of homomorphisms} if $s' \in \ser{S}{\Gamma^*} \cap K$ and $h:\Gamma^*\rightarrow \Sigma^*$ a homomorphism implies $h^{-1}(s') \in K$. The same notions are used for weighted tree languages.

The class of weighted languages recognized by \wfa is closed under (inverses of) homomorphisms (Lemma 6.2 in \cite{droste2013weighted}). \wafa are also closed under inverses of homomorphisms. In fact, this is an easy corollary of Lemma \ref{lmm:d1WafaIffWfta} and Lemma \ref{lmm:d2WafaIffWfta}.

\begin{crl}\label{crl:ClIHWafa} The class of weighted languages recognized by \wafa is closed under inverses of homomorphisms.
\end{crl}

\begin{proof} Let $h':\Lambda^*\rightarrow \Sigma^*$ be a homomorphism and $s\in \ser{S}{\Sigma^*}$ recognized by a \wafa. Due to Lemma \ref{lmm:d1WafaIffWfta}, we get $s=\bhv{B}\circ h^r$. Clearly, $h^r\circ h':\Lambda^*\rightarrow \Sigma^r_\#$ is a tree homomorphism. Thus, due to Lemma \ref{lmm:d2WafaIffWfta}, ${h'}^{-1}(s)=(\bhv{B}\circ h^r)\circ h'=\bhv{B}\circ (h^r\circ h')$ is recognized by a \wafa.
\end{proof}

However, the same is not true for the closure under homomorphisms.

\begin{lmm}\label{lmm:NClHWafa} The class of weighted languages recognized by \wafa is not closed under homomorphisms.
\end{lmm}

\begin{proof} Let $\Sigma=\{a,b,\#\}$, $\mathbb{B}$ the Boolean semiring, and $\mathbb{B}[x]$ the semiring of polynomials in one indeterminate. Consider \[r_B: \Sigma^*\rightarrow \mathbb{B}[x]: w \mapsto \begin{cases}\sum\limits_{k=0}^{j} x^{ki} &\text{if } w=a^i\#b^j \enspace, \\ 0 &\text{otherwise.}\end{cases}\] Due to Lemma 8.3 from \cite{KOSTOLANYI20181}, we know $r_B$ is not recognized by a \wafa. Let $\Gamma=\{a,c,d,\#\}$, $h:\Gamma^*\rightarrow \Sigma^*$ the non-deleting homomorphism induced by $h(a)=a,h(\#)=\#, h(c)=h(d)=b$, and \[r_R: \Gamma^*\rightarrow \mathbb{B}[x]:w \mapsto \begin{cases} x^{ki} &\text{if } w=a^i\#c^kd^l \enspace, \\ 0 &\text{otherwise.}\end{cases}\] Then \[\begin{array}{llll}
h(r_R)(w) &= \sum\limits_{v\in h^{-1}(w)} r_R(v) &=\begin{cases}\sum\limits_{v\in h^{-1}(w)} r_R(v) &\text{if } w=a^i\#b^j \\ 0 &\text{otherwise}\end{cases} &\\
&=\begin{cases}\sum\limits_{k=0}^{j} r_R(a^i\#c^kd^{j-k}) &\text{if } w=a^i\#b^j \\ 0 &\text{otherwise}\end{cases} &=\begin{cases}\sum\limits_{k=0}^{j} x^{ki} &\text{if } w=a^i\#b^j \\ 0 &\text{otherwise}\end{cases} &=r_B(w) \end{array}\] for all $w \in \Sigma^*$. Thus, $h(r_R)$ is not recognized by a \wafa. The series $r_R$ is recognized by the \wafa $\A_R$ which can be found in the appendix of the long version. This completes our proof.
\end{proof}

Nonetheless, the proof of the second direction of Theorem \ref{NfWFA} relies on the closure under homomorphisms. Thus, due to Lemma \ref{lmm:NClHWafa}, a one to one translation of Theorem \ref{NfWFA} into the framework of alternating automata is prohibited. Moreover, in the proof of the first direction of Theorem \ref{NfWFA}, $L$ is defined as a language of runs of $\A$. As mentioned above, runs of \wafa are trees. Therefore, we will utilize a Nivat-like theorem for \wfta to prove the corresponding result for \wafa.

\subsection{A Nivat-like characterization of \wfta}
Nivat-like characterizations for weighted tree languages have been investigated in the past. Unranked trees were considered in \cite{DBLP:conf/birthday/DrosteG17}, while a very general result for graphs can be found in \cite{10.1007/978-3-662-48057-1_15}. Here, for the readers convenience, we want to restate a more restricted version for ranked trees. 

Let $h:\tlg\rightarrow T_\Lambda$ be a non-deleting tree homomorphism, $s\in \ser{S}{T_\Lambda}$. In analogy to words, we define $h(s) \in \ser{S}{T_\Gamma}$ by \(h(s)(t) = \sum_{t' \in h^{-1}(t)} s(t')\) for all $t \in T_\Gamma$. For linear homomorphisms the following is known: 

\begin{lmm}[Theorem 3.8 in \cite{HWAc9}]\label{lmm:TcluioH} The class of weighted tree languages recognized by \wfta is closed under linear homomorphisms.
\end{lmm}

Based on this, it is easy to prove the following result:

\begin{thr}[Nivat-like theorem for \wfta (Theorem 12 in \cite{DBLP:conf/birthday/DrosteG17})]\label{thr:nivatWFTA} A weighted tree language $s \in \ser{S}{T_\Gamma}$ is recognized by a \wfta if and only if there exist a ranked alphabet $\Lambda$, a linear tree homomorphism $h:T_\Lambda\rightarrow T_\Gamma$, a regular tree language $L \subseteq T_\Lambda$, and a \wfta $\A_w$ with exactly one state such that:\[s = h(\bhv{\A_w} \odot \chr{L}) \enspace.\]
\end{thr}

\begin{proof}  A proof can be found in \cite{DBLP:conf/birthday/DrosteG17}. There, the $\Rightarrow$-direction is proved based on a \wfta $\A$ recognizing $s$. The components $\Lambda$, $h$, $L$, and $\A_w$ are chosen to be the set of transition in $\A$, the mapping to the letters in $\Gamma$, the tree language of runs in $\A$, and an automaton adding weights to the transitions, respectively. This yields the desired equation. See Appendix in the long version for a proof of correctness of this construction.
\end{proof}

Based on this result and Theorem \ref{thr:WafaIffWfta} a characterization of \wafa via a Nivat-like Theorem is immediate.

\begin{thr}[Nivat-like theorem for \wafa]\label{thr:nivatWAFA} A weighted language $s \in \ser{S}{\Sigma}$ is recognized by a \wafa if and only if there exist a rank $r\in \N$, a ranked alphabet $\Lambda$, a linear tree homomorphism $h:T_\Lambda\rightarrow T_{\Sigma^r_\#}$, a regular tree language $L \subseteq T_\Lambda$, and a \wfta $\A_w$ with exactly one state. And for all $w \in \Sigma^*$, it holds:\[s(w) = h(\bhv{\A_w} \odot \chr{L})(\tw) \enspace.\]
\end{thr}
\begin{proof}
$\Rightarrow$: Let $\A$ be a nice, equalized \wafa such that $\bhv{\A}=s$. Due to Lemma \ref{lmm:d1WafaIffWfta}, $r\in\N$ and a \wfta $\B$ exist such that $s(w)=\bhv{\B}\circ h^r(w)=\bhv{\B}(t_w^r)$. Applying Theorem \ref{thr:nivatWFTA} to $\bhv{\B}$ gives us the desired result.

$\Leftarrow$: By Theorem \ref{thr:nivatWFTA}, there exists a \wfta $\B$ such that $\bhv{\B} = h(\bhv{\A_w}\odot \chr{L})$. Let $h^r:\Sigma^*\rightarrow T_{\Sigma_\#^r}$ be the generic homomorphism. In consequence of Theorem \ref{thr:WafaIffWfta}, a \wafa $\A$ exists such that $\bhv{\A}(w)= \bhv{\B}(h^r(w))=h(\bhv{\A_w}\odot\chr{L})(h^r(w))=h(\bhv{\A_w}\odot\chr{L})(t_w^r)$ for all $w\in \Sigma$. This finishes our proof.
\end{proof}

\section{A logical characterization of \wafa}
Based on Theorem \ref{thr:WafaIffWfta} we are able to give a logical characterization of \wafa (Theorem \ref{thr:WafaIffsrMSO}). For this purpose, we will use the logical characterization by weighted MSO logic for trees which was introduced in \cite{DROSTE2006228}. 

Weighted MSO logic over trees is an extension of MSO logic over trees. It allows for the use of usual MSO formulas, but also incorporates quantitative aspects such as semiring elements and operations, as well as weighted quantifiers. In the end, every weighted MSO formula defines a weighted tree language. More precisely, let $\Gamma$ be a ranked alphabet, each weighted MSO formula $\varphi \in \textup{MSO}(\Gamma,S)$ defines a weighted tree language $\bhv{\varphi}:T_\Gamma \rightarrow S$. Weighted MSO logic is strictly more expressive than \wfta. Nevertheless, it is possible restrict the syntax of weighted MSO in such a way that it characterizes weighted tree languages recognized by \wfta. This fragment is called weighted syntactically restricted MSO (srMSO). Due to a lack of space, we have to omit the formal definition of srMSO. We will use syntax and semantics of weighted srMSO without any changes and refer the interested reader to \cite{HWAc9} or \cite{droste2011weighted}. Our characterization of \wafa will be fully based on the following characterization theorem for \wfta:

\begin{thr}[Theorem 3.49 (A) in \cite{HWAc9}]\label{thr:wftaIffMSO} A weighted tree language $s\in \ser{S}{T_\Gamma}$ is recognized by a \wfta if and only if there exist $\varphi \in \textup{srMSO}(\Gamma,S)$ such that $s=\bhv{\varphi}$.
\end{thr}

However, we still have to handle the homomorphism used in Theorem \ref{thr:WafaIffWfta}. This will be done by choosing an appropriate way of representing words as relational structures.

By definition $\bhv{\varphi}\in \ser{S}{T_\Gamma}$ for all  $\varphi \in \textup{srMSO}(\Gamma,S)$. However, we want to use weighted srMSO on trees to define weighted languages on words. To this end, we define $\bhv{\varphi}_\Sigma\in \ser{S}{T_\Gamma}$ by $\bhv{\varphi}_\Sigma(w)=\bhv{\varphi}({t_w^{\rank(\Gamma)}})$ for all $\varphi \in \textup{srMSO}(\Gamma,S)$, $w\in \Sigma^*$. Since $\textup{srMSO}(\Gamma,S)\subseteq \textup{srMSO}(\Gamma\cup \Sigma_\#^{\rank(\Gamma)},S)$, we can assume without loss of generality $\varphi\in \textup{srMSO}(\Gamma\cup \Sigma_\#^{\rank(\Gamma)},S)$. Hence, $\bhv{\varphi}_\Sigma$ is well defined for all $\Sigma$. It is easy to see that $\bhv{\varphi}_\Sigma= \bhv{\varphi}\circ h$ where $h:\Sigma \rightarrow \Sigma_\#^{\rank(\Gamma)} \cup \Gamma$  is the generic homomorphism.

\begin{thr}\label{thr:WafaIffsrMSO} A weighted language $s\in \ser{S}{\Sigma}$ is recognized by a \wafa if and only if there exist a ranked alphabet $\Gamma$ and $\varphi \in \textup{srMSO}(\Gamma,S)$ such that $s=\bhv{\varphi}_\Sigma$.
\end{thr}

\begin{proof} $\Rightarrow$: Assume $s\in \ser{S}{\Sigma}$ is recognized by a \wafa. By \ref{lmm:d1WafaIffWfta}, there exists $r\in \mathbb{N}$ and a \wfta $\mathcal{B}$ such that $s=\bhv{\mathcal{B}}\circ h^r$. By Theorem \ref{thr:wftaIffMSO}, $\varphi \in \textup{srMSO}(\Sigma_\#^r,S)$ exists such that $\bhv{\mathcal{B}}=\bhv{\varphi}$. Thus $s=\bhv{\mathcal{B}}\circ h^r=\bhv{\varphi}\circ h^r=\bhv{\varphi}_\Sigma$.

$\Leftarrow$: If $s=\bhv{\varphi}_\Sigma$ for some $\varphi\in \textup{srMSO}(\Gamma\cup \Sigma_\#^{\rank(\Gamma)},S)$, we get $s = \bhv{\varphi}\circ h^{\rank(\Gamma)}$. By Theorem \ref{thr:wftaIffMSO}, a \wfta $\mathcal{B}$ exists such that $\bhv{\varphi}=\bhv{\mathcal{B}}$. Therefore $s=\bhv{\mathcal{B}}\circ h^{\rank(\Gamma)}$. Since $\mathcal{B}$ is a \wfta and $h^{\rank(\Gamma)}$ a homomorphism, a \wafa $\mathcal{A}$ with $s=\bhv{A}$ exists by Lemma \ref{lmm:d2WafaIffWfta}.
\end{proof}

\section{Closure of \wfta under inverses of homomorphisms}
It is a well known fact that regular tree languages are closed under inverses of homomorphisms. Sadly, this is not true in the weighted case, at least not for arbitrary semirings. This raises the question if it is possible to give a precise description of the class of semirings $\Se$ for which \wfta are closed under inverses of homomorphisms. This question will be answered by Theorem \ref{thr:clOTAUIHom}.

Theorem \ref{thr:nivatWAFA} and Theorem \ref{thr:wftaIffMSO} used Theorem \ref{thr:WafaIffWfta} to apply known results for \wfta to \wafa. Vice versa, we can use Theorem \ref{thr:WafaIffWfta} and Theorem \ref{thr:WAFAiffWFAiffLocFin} to characterize $\Se$.

\begin{thr}\label{thr:clOTAUIHom} The class of $S$-weighted tree languages recognized by $\wfta$ is closed under inverses of homomorphisms if and only if $S$ is locally finite.
\end{thr}

To prove this result we will use the notion of \textit{recognizable step functions}: A weighted tree language $r\in \ser{S}{T_\Gamma}$ is a recognizable step function if there exist recognizable tree languages $L_1,\ldots, L_k$ and values $l_1,\ldots, l_k\in S$ such that $r=\sum_{i=1}^k l_i\cdot\chr{L_i}$. Due to \cite{DROSTE2006228}, we know the following about recognizable step functions:

\begin{lmm}[Lemma 3.1 in \cite{DROSTE2006228}]\label{lmm:auxfD1} If $r\in \ser{S}{T_\Gamma}$ is a recognizable step function, then a partition $L_1,\ldots, L_k$ of $T_\Gamma$ exits such that $r=\sum_{i=1}^k l_i\cdot\chr{L_i}$ for some $l_1,\ldots,l_k \in S$.
\end{lmm}

Note, this lemma is not redundant since the definition of recognizable step functions does not demand that the recognizable tree languages are pairwise disjoint. Due to Lemma \ref{lmm:auxfD1}, we know that a weighted tree language is a recognizable step function if and only if it has a finite image and each preimage is a recognizable tree language. The next Lemma characterizes recognizable weighted tree languages over locally finite semirings.

\begin{lmm}[Lemma 3.3 \& Lemma 6.1 in \cite{DROSTE2006228}]\label{lmm:auxfD2} Let $S$ be locally finite. A weighted tree language $r\in \ser{S}{T_\Gamma}$ is recognizable if and only if $r$ is a recognizable step function.
\end{lmm}

Finally, we can proceed with the proof of Theorem \ref{thr:clOTAUIHom}.

\begin{proof}[Proof of Theorem \ref{thr:clOTAUIHom}]$\Rightarrow$: Assume the class of $S$-weighted tree languages recognized by $\wfta$ is closed under inverses of homomorphisms. 

\begin{clm}\label{clm:hClPauxfD2} The class of $S$-weighted $\Sigma$ languages recognizable by \wafa and the class of $S$-weighted $\Sigma$ languages recognizable by \wfa are equal.
\end{clm}

Clearly every \wfa is a \wafa. Thus, for the proof of the claim, it remains to show that every weighted language which is recognized by a \wafa is recognized by a \wfa. For this purpose, assume $s\in \ser{S}{\Sigma^*}$ is recognized by a \wafa $\A$. Due to Theorem \ref{thr:WafaIffWfta}, there exists a \wfta $\B$ and a homomorphism $h: \Sigma^* \rightarrow T_\Gamma$ such that $\bhv{\A}=\bhv{B}\circ h$. However, by our assumption there exists a \wfta $\mathcal{C}$ over $\Sigma_\#^1$ such that $\bhv{C}=\bhv{B}\circ h$ and hence a \wfa $\mathcal{C}'$ over $\Sigma$ such that $\bhv{\A}=\bhv{B}\circ h=\bhv{C}=\bhv{C'}$. Thereby, $s$ is recognized by a \wfa. This proves our claim.

By Claim \ref{clm:hClPauxfD2} and Theorem \ref{thr:WAFAiffWFAiffLocFin} it follows that $S$ is locally finite.

$\Leftarrow$: Assume $S$ is locally finite. Let $r\in \ser{S}{T_\Gamma}$ be recognizable and $h:T_\Lambda \rightarrow T_\Gamma$ a homomorphism. Due to Lemma \ref{lmm:auxfD1} and Lemma \ref{lmm:auxfD2}, we have $r=\sum_{i=1}^k l_i\cdot \chr{L_i}$ for some partition $L_1,\ldots,L_k$ of $T_\Gamma$ and values $l_1,\ldots,l_k \in S$. We claim $r\circ h = \sum_{i=1}^k l_i\cdot \chr{h^{-1}(L_i)}$. To prove this, consider some arbitrary $t\in T_\Lambda$. Since the $L_i$ form a partition of $T_\Gamma$ there exists a unique $j\in\{1,\ldots, k\}$ such that $h(t) \in L_i$. Therefore we have\[\begin{array}{cl}&(r\circ h)(t)\\=&\sum_{i=1}^k l_i\cdot\chr{L_i}(h(t))=l_j\\=&l_j\cdot \chr{h^{-1}(L_j)}(t)\\\overset{l_j \text{ unique}}{=}&\sum_{i=1}^k l_i\cdot \chr{h^{-1}(L_i)} (t)\enspace.  \end{array}\] Since recognizable tree languages are closed under inverses of homomorphisms, we know that $h^{-1}(L_1), \ldots,\\ h^{-1}(L_k) \subseteq T_{\Lambda}$ are recognizable. Thus, $r\circ h$ is a recognizable step function. Again, by Lemma \ref{lmm:auxfD2}, we get $r\circ h$ is recognizable. This completes our proof.
\end{proof}

\section{\wafa and polynomial automata}
We will use known results for polynomial automata, to prove the decidability of the Zeroness Problem for \wafa if weights are taken from the rationals (Lemma \ref{crl:DecZerEqWafa}).

Polynomial automata where introduced in \cite{8005101} as a generalization of both vector addition systems and weighted automata. Polynomial automata are quite similar to \wafa, the authors of \cite{8005101} even prove that the characteristic function of the reversal of each language recognized by a non-weighted alternating automaton is recognized by a polynomial automaton of the same size. We want to strengthen this connection. In \cite{8005101} polynomial automata are defined over the rational numbers. However, it is easy to give a more general definition for arbitrary commutative semirings.

A \textit{polynomial automaton} (\pola) is a $5$-tuple $\A=(n,\Sigma, \alpha, p, \gamma)$, where $n\in \N$ is the number of states, $\Sigma$ is an alphabet, $\alpha \in S^n$ is an initial weight vector, $p:\Sigma\rightarrow {\pols{X_n}}^n$ the transition function, and $\gamma\in \pols{X_n}$ an output polynomial. We denote the $i$-th entry of $p(a)$ by $p_i(a)$.

Let $\A=(n,\Sigma, \alpha, p, \gamma)$ be a \pola. Its \textit{state behavior} $\bhvs{\A}:\{1,\ldots, n\}\times \Sigma ^*\rightarrow S$ is the mapping defined by \[\bhvs{\A}(i,w) = \begin{cases}
  \alpha_i &\text{ if } w=\varepsilon,\\
  p_i\big\langle\bhvs{\A}(1,v),\ldots,\bhvs{\A}(n,v)\big\rangle &\text{ if } w=va \text{ for } a\in \Sigma.
\end{cases}\]
Usually we will denote $\bhvs{\A}(i,w)$ by $\bhvs{\A}_i(w)$. Now, the \textit{behavior} of $\A$ is the weighted language $\bhv{A}:\Sigma^*\rightarrow S$ defined by \[\bhv{\A}(w) = \gamma\big(\bhvs{\A}_1(w),\ldots,\bhvs{\A}_n(w)\big).\]

It is easy to check that this definition is a reformulation of the definition found in \cite{8005101}.

Let the reversal of a weighted language $s\in \ser{S}{\Sigma}$ be defined by $s^R(w)=s(w^R)$ for all $w=w_1,\ldots w_n\in \Sigma^*$, where $w^R=w_n\ldots w_1$. Comparing the definition of state behavior for \wafa and \pola already yields the following lemma:

\begin{lmm}\label{lmm:WafaIffRPola} A weighted language $s\in \ser{S}{\Sigma}$ is recognized by a \wafa if and only if $s^R$ is recognized by a \pola.
\end{lmm}

\begin{proof} Assume $s$ is recognized by $\A=\wafaa$. Let $\B=(|Q|,\Sigma, \big(\tau(q_1),\ldots,\tau(q_n)\big), p, P_0)$ be a \pola with $p_i(a)=\delta(q_i,a)$ for all $1\leq i\leq |Q|, a\in \Sigma$. Then, a straightforward induction on $|w|$ shows $\bhv{A}(w)=\bhv{\B}(w^R)$ for all $w\in \Sigma$.
The second direction is proven analogously to the first one.
\end{proof}

Let $\A,\A'$ be two \wafa. We observe $\bhv{\A}(w) = 0$ for all $w\in \Sigma^*$ if and only if $\bhv{\A}(w^R)=0$ for all $w \in \Sigma^*$. Moreover, we have $\bhv{\A}(w) = \bhv{\A'}(w)$ for all $w\in \Sigma^*$ if and only if $\bhv{\A}(w^R)=\bhv{\A'}(w^R)$ for all $w \in \Sigma^*$. This allows us to derive the following corollary from Lemma \ref{lmm:WafaIffRPola}:

\begin{crl}\label{crl:DecZerEqWafa} The Zeroness Problem and the Equivalence Problem for \wafa with weights taken from the rationals are in the complexity class ACKERMANN and hard for the complexity class ACKERMANN.
\end{crl}

\begin{proof} Lemma \ref{lmm:WafaIffRPola} + Theorem 1, Theorem 4, Corollary 1 in \cite{8005101}.
\end{proof}

\section{Conclusion}
We were able to connect WAFA to a variety of formalisms, giving a better understanding of their expressive power and characterizing the class of quantitative languages recognized by WAFA. From here, there are various routes to take. It could be of great practical use to find a logical characterization of WAFA via a linear formalism such as a weighted linear logic, or weighted rational expressions tailored to the expressive power of WAFA. Similar to the work in \cite{8005101}, one could investigate subclasses of WAFA allowing for more efficient decision procedures. Alternatively, one could approach the concept of alternation in weighted automata dealing with more complex structures than words, such as weighted alternating tree automata. And of course, having the universal interpretation of nondeterminism in mind, one may take several of these routes at once!

\nocite{*}
\bibliographystyle{eptcs}
\bibliography{submission34}

\begin{thebibliography}{10}
\providecommand{\bibitemdeclare}[2]{}
\providecommand{\surnamestart}{}
\providecommand{\surnameend}{}
\providecommand{\urlprefix}{Available at }
\providecommand{\url}[1]{\texttt{#1}}
\providecommand{\href}[2]{\texttt{#2}}
\providecommand{\urlalt}[2]{\href{#1}{#2}}
\providecommand{\doi}[1]{doi:\urlalt{http://dx.doi.org/#1}{#1}}
\providecommand{\eprint}[1]{ArXiv:\urlalt{https://arxiv.org/abs/#1}{#1}}
\providecommand{\bibinfo}[2]{#2}

\bibitemdeclare{inproceedings}{10.1007/978-3-642-24372-1_2}
\bibitem{10.1007/978-3-642-24372-1_2}
\bibinfo{author}{S.~\surnamestart Almagor\surnameend} \&
  \bibinfo{author}{O.~\surnamestart Kupferman\surnameend}
  (\bibinfo{year}{2011}): \emph{\bibinfo{title}{Max and sum Semantics for
  alternating weighted automata}}.
\newblock In \bibinfo{editor}{T.~\surnamestart Bultan\surnameend} \&
  \bibinfo{editor}{P.~\surnamestart Hsiung\surnameend}, editors: {\sl
  \bibinfo{booktitle}{Automated Technology for Verification and Analysis}},
  \bibinfo{publisher}{Springer}, pp. \bibinfo{pages}{13--27},
  \doi{10.1007/978-3-642-24372-1_2}.

\bibitemdeclare{inproceedings}{8005101}
\bibitem{8005101}
\bibinfo{author}{M.~\surnamestart Benedikt\surnameend}, \bibinfo{author}{Duffm
  \surnamestart T.\surnameend}, \bibinfo{author}{A.~\surnamestart
  Sharad\surnameend} \& \bibinfo{author}{J.~\surnamestart Worrell\surnameend}
  (\bibinfo{year}{2017}): \emph{\bibinfo{title}{Polynomial automata: Zeroness
  and applications}}.
\newblock In: {\sl \bibinfo{booktitle}{ACM/IEEE Symposium on Logic in Computer
  Science (LICS)}}, pp. \bibinfo{pages}{1--12},
  \doi{10.1109/LICS.2017.8005101}.
\newblock
  \urlprefix\url{https://ora.ox.ac.uk/objects/uuid:f341421b-2130-42d7-bb21-52442bad0b80#citeForm}.

\bibitemdeclare{article}{BRZOZOWSKI198019}
\bibitem{BRZOZOWSKI198019}
\bibinfo{author}{J.A. \surnamestart Brzozowski\surnameend} \&
  \bibinfo{author}{E.~\surnamestart Leiss\surnameend} (\bibinfo{year}{1980}):
  \emph{\bibinfo{title}{On equations for regular languages, finite automata,
  and sequential networks}}.
\newblock {\sl \bibinfo{journal}{Theoretical Computer Science}}
  \bibinfo{volume}{10}(\bibinfo{number}{1}), pp. \bibinfo{pages}{19--35},
  \doi{10.1016/0304-3975(80)90069-9}.
\newblock
  \urlprefix\url{https://www.sciencedirect.com/science/article/pii/0304397580900699}.

\bibitemdeclare{article}{alternation}
\bibitem{alternation}
\bibinfo{author}{A.~K. \surnamestart Chandra\surnameend},
  \bibinfo{author}{D.~C. \surnamestart Kozen\surnameend} \&
  \bibinfo{author}{L.~J. \surnamestart Stockmeyer\surnameend}
  (\bibinfo{year}{1981}): \emph{\bibinfo{title}{Alternation}}.
\newblock {\sl \bibinfo{journal}{J. ACM}}
  \bibinfo{volume}{28}(\bibinfo{number}{1}), pp. \bibinfo{pages}{114--133},
  \doi{10.1145/322234.322243}.

\bibitemdeclare{inproceedings}{Chatterjee2008}
\bibitem{Chatterjee2008}
\bibinfo{author}{K.~\surnamestart Chatterjee\surnameend},
  \bibinfo{author}{L.~\surnamestart Doyen\surnameend} \& \bibinfo{author}{T.~A.
  \surnamestart Henzinger\surnameend} (\bibinfo{year}{2008}):
  \emph{\bibinfo{title}{Quantitative languages}}.
\newblock In \bibinfo{editor}{M.~\surnamestart Kaminski\surnameend} \&
  \bibinfo{editor}{S.~\surnamestart Martini\surnameend}, editors: {\sl
  \bibinfo{booktitle}{Computer Science Logic}}, \bibinfo{publisher}{Springer},
  pp. \bibinfo{pages}{385--400}, \doi{10.1007/978-3-540-87531-4_28}.

\bibitemdeclare{inproceedings}{10.1007/978-3-642-03409-1_2}
\bibitem{10.1007/978-3-642-03409-1_2}
\bibinfo{author}{K.~\surnamestart Chatterjee\surnameend},
  \bibinfo{author}{L.~\surnamestart Doyen\surnameend} \& \bibinfo{author}{T.~A.
  \surnamestart Henzinger\surnameend} (\bibinfo{year}{2009}):
  \emph{\bibinfo{title}{Alternating Weighted Automata}}.
\newblock In \bibinfo{editor}{M.~\surnamestart Kuty{\l}owski\surnameend},
  \bibinfo{editor}{W.~\surnamestart Charatonik\surnameend} \&
  \bibinfo{editor}{M.~\surnamestart G{\k{e}}bala\surnameend}, editors: {\sl
  \bibinfo{booktitle}{Fundamentals of Computation Theory}},
  \bibinfo{publisher}{Springer}, pp. \bibinfo{pages}{3--13},
  \doi{10.1007/978-3-642-03409-1_2}.

\bibitemdeclare{article}{MULLER1987267}
\bibitem{MULLER1987267}
\bibinfo{author}{E.~Muller \surnamestart D\surnameend} \&
  \bibinfo{author}{P.~E. \surnamestart Schupp\surnameend}
  (\bibinfo{year}{1987}): \emph{\bibinfo{title}{Alternating automata on
  infinite trees}}.
\newblock {\sl \bibinfo{journal}{Theoretical Computer Science}}
  \bibinfo{volume}{54}(\bibinfo{number}{2}), pp. \bibinfo{pages}{267--276},
  \doi{10.1016/0304-3975(87)90133-2}.
\newblock
  \urlprefix\url{https://www.sciencedirect.com/science/article/pii/0304397587901332}.

\bibitemdeclare{inproceedings}{DeGiacomo:2013:LTL:2540128.2540252}
\bibitem{DeGiacomo:2013:LTL:2540128.2540252}
\bibinfo{author}{G.~\surnamestart De~Giacomo\surnameend} \&
  \bibinfo{author}{M.~Y. \surnamestart Vardi\surnameend}
  (\bibinfo{year}{2013}): \emph{\bibinfo{title}{Linear temporal logic and
  linear dynamic logic on finite traces}}.
\newblock In: {\sl \bibinfo{booktitle}{Proceedings of the Twenty-Third
  International Joint Conference on Artificial Intelligence}},
  \bibinfo{publisher}{AAAI Press}, pp. \bibinfo{pages}{854--860},
  \doi{10.5555/2540128.2540252}.
\newblock \urlprefix\url{http://dl.acm.org/citation.cfm?id=2540128.2540252}.

\bibitemdeclare{inproceedings}{10.1007/978-3-662-48057-1_15}
\bibitem{10.1007/978-3-662-48057-1_15}
\bibinfo{author}{M.~\surnamestart Droste\surnameend} \&
  \bibinfo{author}{S.~\surnamestart D{\"u}ck\surnameend}
  (\bibinfo{year}{2015}): \emph{\bibinfo{title}{Weighted automata and logics on
  graphs}}.
\newblock In \bibinfo{editor}{G.~F. \surnamestart Italiano\surnameend},
  \bibinfo{editor}{G.~\surnamestart Pighizzini\surnameend} \&
  \bibinfo{editor}{D.~T. \surnamestart Sannella\surnameend}, editors: {\sl
  \bibinfo{booktitle}{Mathematical Foundations of Computer Science 2015}},
  \bibinfo{publisher}{Springer}, pp. \bibinfo{pages}{192--204},
  \doi{10.1007/978-3-662-48057-1_15}.

\bibitemdeclare{article}{DROSTE200769}
\bibitem{DROSTE200769}
\bibinfo{author}{M.~\surnamestart Droste\surnameend} \&
  \bibinfo{author}{P.~\surnamestart Gastin\surnameend} (\bibinfo{year}{2007}):
  \emph{\bibinfo{title}{Weighted automata and weighted logics}}.
\newblock {\sl \bibinfo{journal}{Theoretical Computer Science}}
  \bibinfo{volume}{380}(\bibinfo{number}{1}), pp. \bibinfo{pages}{69--86},
  \doi{10.1016/j.tcs.2007.02.055}.
\newblock
  \urlprefix\url{https://www.sciencedirect.com/science/article/pii/S0304397507001582}.
\newblock \bibinfo{note}{Automata, Languages and Programming}.

\bibitemdeclare{inbook}{HWAc5}
\bibitem{HWAc5}
\bibinfo{author}{M.~\surnamestart Droste\surnameend} \&
  \bibinfo{author}{P.~\surnamestart Gastin\surnameend} (\bibinfo{year}{2009}):
  \emph{\bibinfo{title}{Weighted Automata and Weighted Logics}},
  chapter~\bibinfo{chapter}{5}.
\newblock ~\bibinfo{volume}{1} of \bibinfo{editor}{Droste} et~al.
  \cite{Droste:2009:HWA:1667106}, \doi{10.1007/978-3-642-01492-5_5}.

\bibitemdeclare{inproceedings}{DBLP:conf/birthday/DrosteG17}
\bibitem{DBLP:conf/birthday/DrosteG17}
\bibinfo{author}{M.~\surnamestart Droste\surnameend} \&
  \bibinfo{author}{D.~\surnamestart G{\"{o}}tze\surnameend}
  (\bibinfo{year}{2017}): \emph{\bibinfo{title}{A Nivat theorem for
  quantitative automata on unranked trees}}.
\newblock In: {\sl \bibinfo{booktitle}{Models, Algorithms, Logics and Tools -
  Essays Dedicated to Kim Guldstrand Larsen on the Occasion of His 60th
  Birthday}}, pp. \bibinfo{pages}{22--35}, \doi{10.1007/978-3-319-63121-9\_2}.

\bibitemdeclare{book}{Droste:2009:HWA:1667106}
\bibitem{Droste:2009:HWA:1667106}
\bibinfo{editor}{M.~\surnamestart Droste\surnameend},
  \bibinfo{editor}{W.~\surnamestart Kuich\surnameend} \&
  \bibinfo{editor}{H.~\surnamestart Vogler\surnameend}, editors
  (\bibinfo{year}{2009}): \emph{\bibinfo{title}{Handbook of Weighted
  Automata}}, \bibinfo{edition}{1st} edition.
\newblock \bibinfo{volume}{1}, \bibinfo{publisher}{Springer},
  \doi{10.1007/978-3-642-01492-5}.

\bibitemdeclare{inbook}{droste2013weighted}
\bibitem{droste2013weighted}
\bibinfo{author}{M.~\surnamestart Droste\surnameend} \&
  \bibinfo{author}{D.~\surnamestart Kuske\surnameend} (\bibinfo{year}{2021}):
  \emph{\bibinfo{title}{Weighted automata}}, pp. \bibinfo{pages}{113–--150}.
\newblock ~\bibinfo{volume}{1} of \bibinfo{editor}{Pin}
  \cite{realdroste2013weighted}, \doi{10.4171/Automata-1/4}.

\bibitemdeclare{article}{DROSTE2006228}
\bibitem{DROSTE2006228}
\bibinfo{author}{M.~\surnamestart Droste\surnameend} \&
  \bibinfo{author}{H.~\surnamestart Vogler\surnameend} (\bibinfo{year}{2006}):
  \emph{\bibinfo{title}{Weighted tree automata and weighted logics}}.
\newblock {\sl \bibinfo{journal}{Theoretical Computer Science}}
  \bibinfo{volume}{366}(\bibinfo{number}{3}), pp. \bibinfo{pages}{228 -- 247},
  \doi{10.1016/j.tcs.2006.08.025}.
\newblock
  \urlprefix\url{http://www.sciencedirect.com/science/article/pii/S0304397506005676}.
\newblock \bibinfo{note}{Automata and Formal Languages}.

\bibitemdeclare{article}{droste2011weighted}
\bibitem{droste2011weighted}
\bibinfo{author}{M.~\surnamestart Droste\surnameend} \&
  \bibinfo{author}{H.~\surnamestart Vogler\surnameend} (\bibinfo{year}{2011}):
  \emph{\bibinfo{title}{Weighted logics for unranked tree automata}}.
\newblock {\sl \bibinfo{journal}{Theory of Computing Systems}}
  \bibinfo{volume}{48}(\bibinfo{number}{1}), pp. \bibinfo{pages}{23--47},
  \doi{10.1007/s00224-009-9224-4}.

\bibitemdeclare{inbook}{HWAc9}
\bibitem{HWAc9}
\bibinfo{author}{Z.~\surnamestart F{\"u}l{\"o}p\surnameend} \&
  \bibinfo{author}{H.~\surnamestart Vogler\surnameend} (\bibinfo{year}{2009}):
  \emph{\bibinfo{title}{Weighted Tree Automata and Tree Transducers}},
  chapter~\bibinfo{chapter}{9}.
\newblock ~\bibinfo{volume}{1} of \bibinfo{editor}{Droste} et~al.
  \cite{Droste:2009:HWA:1667106}, \doi{10.1007/978-3-642-01492-5_9}.

\bibitemdeclare{article}{KOSTOLANYI20181}
\bibitem{KOSTOLANYI20181}
\bibinfo{author}{P.~\surnamestart Kostolányi\surnameend} \&
  \bibinfo{author}{F.~\surnamestart Mišún\surnameend} (\bibinfo{year}{2018}):
  \emph{\bibinfo{title}{Alternating weighted automata over commutative
  semirings}}.
\newblock {\sl \bibinfo{journal}{Theoret. Comput. Sci.}} \bibinfo{volume}{740},
  pp. \bibinfo{pages}{1 -- 27}, \doi{10.1016/j.tcs.2018.05.003}.
\newblock
  \urlprefix\url{http://www.sciencedirect.com/science/article/pii/S0304397518303098}.

\bibitemdeclare{book}{realVardi1995}
\bibitem{realVardi1995}
\bibinfo{editor}{J.~\surnamestart van Leeuwen\surnameend}, editor
  (\bibinfo{year}{1995}): \emph{\bibinfo{title}{Computer Science Today: Recent
  Trends and Developments}}.
\newblock \bibinfo{volume}{1}, \bibinfo{publisher}{Springer},
  \doi{10.1007/BFb0015232}.

\bibitemdeclare{article}{Nivat}
\bibitem{Nivat}
\bibinfo{author}{M.~\surnamestart Nivat\surnameend} (\bibinfo{year}{1968}):
  \emph{\bibinfo{title}{Transductions des langages de Chomsky}}.
\newblock {\sl \bibinfo{journal}{Annales de l'Institut Fourier}}
  \bibinfo{volume}{18}(\bibinfo{number}{1}), pp. \bibinfo{pages}{339--455},
  \doi{10.5802/aif.287}.
\newblock \urlprefix\url{http://www.numdam.org/articles/10.5802/aif.287/}.

\bibitemdeclare{book}{realdroste2013weighted}
\bibitem{realdroste2013weighted}
\bibinfo{editor}{Jean-Éric \surnamestart Pin\surnameend}, editor
  (\bibinfo{year}{2021}): \emph{\bibinfo{title}{Handbook of Automata Theory}},
  \bibinfo{edition}{1st} edition.
\newblock \bibinfo{volume}{1}, \bibinfo{publisher}{European Mathematical
  Society}, \doi{10.4171/Automata}.

\bibitemdeclare{article}{SCHUTZENBERGER1961245}
\bibitem{SCHUTZENBERGER1961245}
\bibinfo{author}{M.P. \surnamestart Schützenberger\surnameend}
  (\bibinfo{year}{1961}): \emph{\bibinfo{title}{On the definition of a family
  of automata}}.
\newblock {\sl \bibinfo{journal}{Inform. and Control}}
  \bibinfo{volume}{4}(\bibinfo{number}{2}), pp. \bibinfo{pages}{245 -- 270},
  \doi{10.1016/S0019-9958(61)80020-X}.
\newblock
  \urlprefix\url{http://www.sciencedirect.com/science/article/pii/S001999586180020X}.

\bibitemdeclare{inbook}{Vardi1995}
\bibitem{Vardi1995}
\bibinfo{author}{M.~Y. \surnamestart Vardi\surnameend} (\bibinfo{year}{1995}):
  \emph{\bibinfo{title}{Alternating automata and program verification}}, pp.
  \bibinfo{pages}{471--485}.
\newblock ~\bibinfo{volume}{1} of \bibinfo{editor}{van Leeuwen}
  \cite{realVardi1995}, \doi{10.1007/BFb0015261}.

\end{thebibliography}
\end{document}